\documentclass[article]{llncs}

\usepackage{xspace}
\usepackage{microtype}
\usepackage{framed}
\usepackage{xifthen}
\usepackage{tikz}
\usepackage{graphicx}
\usepackage{booktabs} % for professional tables
\usepackage{hyperref}
\usepackage{mathtools}
\usetikzlibrary{calc}

\usepackage{algorithmic}
\usepackage{algorithm}

\usepackage[capitalize,noabbrev]{cleveref}

\usepackage{amsmath}
\usepackage{amssymb}
\usepackage{enumitem}
\usepackage{thm-restate}

\newlength{\RoundedBoxWidth}
\newsavebox{\GrayRoundedBox}
\newenvironment{GrayBox}[1]%
   {\setlength{\RoundedBoxWidth}{.93\textwidth}
    \def\boxheading{#1}
    \begin{lrbox}{\GrayRoundedBox}
       \begin{minipage}{\RoundedBoxWidth}}%
   {   \end{minipage}
    \end{lrbox}
    \begin{center}
    \begin{tikzpicture}%
       \node(Text)[draw=black!20,fill=white,rounded corners,%
             inner sep=2ex,text width=\RoundedBoxWidth]%
             {\usebox{\GrayRoundedBox}};
        \coordinate(x) at (current bounding box.north west);
        \node [draw=white,rectangle,inner sep=3pt,anchor=north west,fill=white] 
        at ($(x)+(6pt,.75em)$) {\boxheading};
    \end{tikzpicture}
    \end{center}}     
\newenvironment{defproblemx}[2][]{\noindent\ignorespaces%
                                \FrameSep=6pt%
                                \parindent=0pt%
                \vspace*{-1.5em}
                \ifthenelse{\isempty{#1}}{%
                  \begin{GrayBox}{\textsc{#2}}%                
                }{%
                  \begin{GrayBox}{\textsc{#2}  parameterized by~{#1}}%  
                }
                \begin{tabular*}{\textwidth}{@{\hspace{.1em}} >{\itshape} p{1.8cm} p{0.8\textwidth} @{}}%        
            }{
                \end{tabular*}%
                \end{GrayBox}%
                \ignorespacesafterend
            }  

\usepackage[textsize=large]{todonotes}

\title{Space Efficient Algorithms for Parameterised Problems}

\author{Sheikh Shakil Akhtar\inst{1} \and Pranabendu Misra\inst{1} \and Geevarghese Philip\inst{1}}

\institute{Chennai Mathematical Institute}

\authorrunning{S.S. Akhtar, P. Misra, G. Philip}
\begin{document}

\maketitle

\begin{abstract}
	We study ``space efficient'' FPT algorithms for graph problems with
limited memory. Let $n$ be the size of the input graph and $k$ be the
parameter. We present algorithms that run in time $f(k)\cdot n^{{\cal
O}(1)}$ and use $g(k)\cdot (\log n)^{{\cal O}(1)}$ working space,
where $f$ and $g$ are functions of $k$ alone, for {\sc $k$-Path}, {\sc
MaxLeaf SubTree} and {\sc Multicut in Trees}. These algorithms are
motivated by big-data settings where very large problem instances must
be solved, and using $n^{O(1)}$ memory is prohibitively expensive.
They are also theoretically interesting, since most of the standard methods tools, such as deleting a large set of vertices or edges, are unavailable, and we must a develop different way to tackle them.

\end{abstract}

%\newpage

%\keywords{Parameterised Algorithms, Long Path, Steiner Tree, MaxLeaf,
%Multicut, Graph Theory, Space-bounded computation}

%\newpage
\section{Introduction}

With the increasing use of big data in practical applications, the field of \emph{space-efficient algorithms} has increased in importance. Traditionally, the time required by an algorithm has been the primary focus of analysis. However, when dealing with large volumes of data we must also pay attention to the working space required to run the algorithm, otherwise it will be impossible to run the algorithm at all. 

A few models of computation for space-efficient algorithms have been
proposed, the most prominent among them being \emph{streaming
algorithms}~\cite{10.1145/237814.237823}. Recently streaming algorithms have been studied
in the parameterised setting~\cite{chitnis_et_al:LIPIcs.IPEC.2019.7}, to solve parameterized versions of NP-hard problems.
This model places strong restrictions on how many times one can access
the input data. Consequently, we get strong lower bounds on the amount
of space required, and this restricts the set of problems that can be
efficiently solved~\cite{ghosh_et_al:LIPIcs.ESA.2024.60} in this
model. In particular, it is well-known that for most problems on graphs
with $n$ vertices, the space required by a streaming algorithm is at
least $n \cdot (\log n)^{{\Omega}(1)}$; this is true even for
something as basic as detecting if the input graph contains a cycle
\cite{10.1145/2627692.2627694}.
%\textcolor{red}{Add a reference for this statement.}. 

A more relaxed model allows for the input to be read as many times as
needed, but restricts the amount of working space that the algorithm
can use. This is the model that we consider in this work, where we
study various problems on graphs. Let $n$ be the number of vertices of
the input graph and let $k$ be the \emph{parameter} which
is---typically---the size of the solution that we are looking for. We
study algorithms that run in time $f(k)\cdot n^{{\cal O}(1)}$ and use
$g(k)\cdot (\log n)^{{\cal O}(1)}$ working space, where $f$ and $g$
are functions of $k$ alone. We say that these algorithms are {\em
Fixed-parameter tractable (FPT)}~\cite{DBLP:books/sp/CyganFKLMPPS15}
and space-efficient. Note that with so little working space, many
standard algorithmic tools and techniques become unavailable. Indeed,
even something as simple as deleting some edges from the graph could
become non-trivial to implement. This is because keeping track of a
large arbitrary set of edges that was deleted, will require more
memory than is allowed by the model. Therefore, developing algorithms
under this model seems to require new ideas and methods.

The class of problems solvable in LOGSPACE, and related complexity classes are very well studied in computer science~\cite{arora2009computational}. 
There has been a lot of work on streaming
algorithms for problems on graphs~\cite{FEIGENBAUM2005207,10.1145/2627692.2627694}, and more recently on parameterised streaming algorithms~\cite{doi:10.1137/1.9781611977912.28}.
Space efficient FPT algorithms have been studied
earlier~\cite{bergougnoux_et_al:LIPIcs.ESA.2023.18,ElberfeldST15,DBLP:conf/mfcs/FafianieK15,CHEN2022104951,10.1007/978-981-97-2340-9_22,BodlaenderFOCS21,BodlaenderESA22,BodlaenderIPEC22,BodlaenderWG24}
for {\sc Vertex Cover}, {\sc $d$-Hitting Set}, {\sc Edge-Dominating
Set}, {\sc Maximal Matching}, {\sc Feedback Vertex Set}, {\sc Path
Contraction}, {\sc List Coloring} and other problems. % and {\sc Cluster Editing}.
Furthermore, a theory of hardness (primarily based on strictly using $g(k)\cdot O(\log n)$-space) is being developed~\cite{ElberfeldST15,BodlaenderFOCS21,BodlaenderESA22,BodlaenderIPEC22,BodlaenderWG24}.

In this paper, we present algorithms that use $g(k)\cdot \log^{O(1)} n$-space and take $f(k) \cdot n^{O(1)}$-time for the following problems.
\\
\noindent\fbox{\begin{minipage}{\textwidth}
{\sc $k$-Path}\\
Input: An undirected graph \(G = (V,E)\) and an
integer \(k\).	

Parameter: \(k\)

Question: Does \(G\) have a path on \(k\) vertices?
\end{minipage}}
\\
%\noindent\fbox{\begin{minipage}{\textwidth}
%{\sc Steiner Tree}\\
%Input: An undirected graph \(G = (V,E)\), \(T \subseteq V\) and an
%integer \(k\).	
%
%Parameter: \(k\)
%
%Question: Does \(G\) have a connected subgraph on at most \(k\) vertices, say
%\(H\) such that \(T \subseteq V(H)\)?
%\end{minipage}}

\noindent\fbox{\begin{minipage}{\textwidth}
{\sc MaxLeaf Subtree}\\
Input: An undirected graph \(G = (V,E)\) and an integer \(k\).	

Parameter: \(k\)

Question: Does \(G\) have a subtree with at least \(k\) leaves?
\end{minipage}}
\\

\noindent\fbox{\begin{minipage}{\textwidth}
{\sc Multicut in Trees}\\
Input: An undirected tree \(T = (V,E), n = | V |\), a collection \(H\)
of \(m\) pairs of nodes in \(T\) and an integer \(k\).	

Parameter: \(k\)

Question: Does \(T\) have an edge subset of size at most \(k\) whose
removal separates each pair of nodes in \(H\)?
\end{minipage}}
\\

We obtain the following results:

\begin{restatable}[]{theorem}{kpath}
\label{th:kpath}
	There is a deterministic algorithm, that solves the \textsc{\(k\)-Path}
	problem, runs in time \(n^{\mathcal{O}(1)} \cdot 2^{k^2}
	\cdot k!\) and uses \(\mathcal{O}(2^{k^2} \cdot k!
	\cdot k \cdot \log n)\) working space.
\end{restatable}

%\begin{restatable}[]{theorem}{stntree}
%\label{th:stntree}
%	There is a constant probability, one-sided error randomised
%	algorithm that solves the \textsc{Steiner Tree}
%	problem, which runs in time \(n^{\mathcal{O}(1)} \cdot 2^{k^2} \cdot
%	k! \cdot k^{k - 2}\) and uses
%	\((\text{log } n)^{\mathcal{O}(1)} \cdot 2^{k^2} \cdot k! \cdot
%	k^{k - 1}\) as working space.
%\end{restatable}

\begin{restatable}[]{theorem}{maxleafsubtree}
\label{th:maxleafsubtree}
	There is a deterministic algorithm which solves the
	\textsc{MaxLeaf Subtree} problem, in
	time \(n^{\mathcal{O}(1)} \cdot 4^k\) and uses
	\(\mathcal{O}(4^k \cdot k \cdot \log n)\)
	working space.
\end{restatable}

\begin{restatable}[]{theorem}{multicutintree}
\label{th:multicutintree}
	There is a deterministic algorithm which solves the
	\textsc{Multicut In Trees}
	problem, in
	time \(n^{\mathcal{O}(1)} \cdot 2^k\) and uses \(\mathcal{O}(2^k \cdot k \cdot
	\log n)\) working space.
\end{restatable}

To the best of our knowledge, these are the first results on the above mentioned
problems in the bounded space setting. While the above problems have
theoretical interests, they are also used in practical applications as
well. For example, the \(k\)-{\sc Path} problem can be used in
studying protein-protein interaction
\cite{10.1093/bioinformatics/btn163}, whereas the {\sc MaxLeaf
Subtree} problem can be used in the study of phylogenetic networks
\cite{10.1109/TCBB.2020.3040910}. While it may be difficult to find an
application of the {\sc Multicut in Trees} problem, as real world
scenario often don't occurs as tree, one can get application of the
more general {\sc Multicut in Graphs} problem. For instance, in
\cite{10.5555/1873601.1873635}, the authors use the {\sc Multicut in
Graphs} problem to study reliability in communication networks.

%Some works we need to cite:
%
%\cite{10.1007/978-3-031-22105-7_23} ``Space Limited Graph Algorithms on
%Big Data"
%
%\cite{10.1007/978-981-97-2340-9_22} ``Space-efficient graph
%kernelizations"

\section{Preliminaries}
For any positive integer, \(i\), we denote the set \(\{1, \ldots,
i\}\) by \([i]\). 
%For a positive integer \(i\), a \(j\)-permutation of
%\([i]\), where \(j \leq i\), is a permutation of any subset of \(j\) elements
%from \([i]\).
We will use standard notations from graph theory.
For a graph \(G=(V, E)\), with \(V\) as the vertex set and \(E\) as
the edge set, we will assume an order on \(V\) and hence on the
neighbourhood of every vertex in \(G\). This will help us in avoiding
many of the issues that can arise due to restriction in working space. Also, 
we use \(n\) to denote \(| V |\). We  consider \(V\) to
be the set \(\{v_1, v_2, \ldots , v_n\}\). For any \(S \subseteq V\),
we use \(G[S]\) to denote the induced subgraph of \(G\) on \(S\). And
we use \(G - S\) to denote \(G[V \setminus S]\). Also, if \(S =
\{v\}\), for some \(v \in V\), then we simply write \(G - v\), instead
of \(G - \{v\}\).
For any vertex \(v \in V\), \(N(v) : = \{w \in V | vw \in E\}\) and \(N[v] :
= N(v) \cup \{v\}\). A path on \(k\) vertices will have length \(k -
1\), i.e., the number of edges in it.

A celebrated result of Reingold~\cite{10.1145/1060590.1060647}, gives a
\(\mathcal{O} (\text{log } n)\) space, polynomial time deterministic algorithm for
undirected st-connectivity, where $s$ and $t$ are two vertices in an input-graph $G$ on $n$ vertices. Let's call it \(\mathcal{A}_{con}\).

\textbf{Colour Coding:}
Colour Coding is an algorithmic technique to detect if a given
input graph has a smaller subgraph, isomorphic to
another graph 
(\cite{10.1007/978-3-642-11269-0_1}, \cite{10.1145/195058.195179}). In other words, given a pattern
\(k\)-vertex ``pattern" graph \(H\) and an \(n\)-vertex input graph
\(G\), the goal is to find a subgraph of \(G\) isomorphic to \(H\)
(\cite{DBLP:books/sp/CyganFKLMPPS15}). While this techniqe can be used
to efficiently detect small subgraphs like paths and cycles, it may not
result in efficient algorithms for arbitrary \(H\); for example, if
\(H\) is a complete graph (the \(k\)-{\sc Clique} problem). 

We implement the technique by the standard method of using a family of \emph{universal hash functions}. In particular, if \(G = (V, E)\) is an input graph with
\(n\) vertices and the vertices are labeled as \(\{v_1, v_2, \cdots ,
v_n\}\), then when we need to find a vertex subset of size \(k\), we
use the following family of hash functions.
Let \(p\) be a prime number greater than \(n\). We define the
following family of hash functions.

\begin{multline}
	\label{eq:hash}
	\mathcal{H}_u := \{h_{a,b} : \{1, \ldots , n\} \rightarrow
		\{1, \ldots , k^2\} \mid 
		a, b \in \{0, \ldots , p - 1\}, a \geq 1,
        %b \in \{0, 1, \ldots , p - 1\},
	\forall i \in [n], \\
    h_{a,b}(i) = ((ai + b) \text{mod } p) \text{mod } k^2\}
\end{multline}

Given \(n\) and \(p\), each of the above functions can be evaluated in
\(\mathcal{O}(n)\) time and \(\mathcal{O}(\text{log } n)\) space.
The set \([k^2]\) is called the set of colours.
Each member of \(\mathcal{H}_u\) is a called a \textit{colouring
function} and given a colouring function, \(h_{a,b}\), for each \(i
\in [n]\), \(h_{a,b}(i)\) is called the colour of the vertex \(v_i\).

A subset \(S\) of \(V\) will be called \textit{properly coloured} or
\textit{colourful} under a given colouring function, if every member
of \(S\) has a distinct colour, i.e., the function is injective when
restricted to \(S\). The family \(\mathcal{H}_u\) has the 
property of being a universal hash family~\cite{CARTER1979143}.
This means that the probability that a given subset of vertices, say
\(S\), where \(| S | \leq k\), is \emph{not} colourful under a
colouring function chosen uniformly at random from the family \(\mathcal{H}_u\), is at most
\(1/2\) (see,
e.g.,~\cite{10.1007/978-3-031-22105-7_23,10.5555/1614191}). Thus,
there is some function from the family \(\mathcal{H}_u\), which when
restricted to \(S\) will be injective.
Given the prime number $p$, $n$ and $k$, we can enumerate the functions in \(\mathcal{H}_u\), lexicographically with respect to the pair \((a, b)\), in time \(\mathcal{O}(n^3)\) and space \(\mathcal{O}(\text{log } n)\), deterministically.
%\todo[inline]{Clarify this probability? What is it over?}

\textbf{How do we get a sufficiently large \(p\), deterministically, given the space restrictions?} 
By Bertrand's Postulate (later Theorem), for any \(n \geq 2\) there exists a prime \(p\), such that \(n < p < 2n\). Therefore, by simply testing the primality of each integer between $n+1$ and $2n-1$, we can obtain a prime $p > n$. 
As we will be dealing with an integer of absolute-value at most \(2n\), thus the number of bits required to represent these integers will be at most \(\text{log } 2n\) which is \(\mathcal{O}(\text{log } n)\). To test the primality of a number $q$, we simply check if any integer $r < q$ divides it, which can be done in ${\cal O}(q \log n)$ time and ${\cal O}(\log q)$ space.
Thus, by simply testing each integer between $n+1$ and $2n-1$, we have a deterministic algorithm that runs in time ${\cal O}(n^2 \log n)$   
%\textcolour{red}{Why polynomial in \(n\)? Isn't the running time also
%polynomial in \(\log{n}\)?} 
and space ${\cal O}(\log n)$, and outputs a prime \(p > n\).

The celebrated AKS primality test by Agrawal et al.~\cite{b68c33ca-3366-3b13-8901-69e76cc88da6} gives a
deterministic algorithm to check if an input integer $q$ is prime or not, where the running time (and hence the space complexity) is $\log^{{\cal O}(1)} q$.
%On the other-hand, the popular Miller-Rabin Primality Test~\cite{rabin1980probabilistic}, is randomized but requires $O(\log q)$ space and $\log^{O(1)} q$ time.
%As we will be dealing with an integer of absolute-value at most \(2n\), thus the number of bits required to represent these integers will be at most \(\text{log } 2n\) which is \(\mathcal{O}(\text{log } n)\). 
%Thus, by simply testing each integer between $n+1$ and $2n$, we have a randomized algorithm that runs in time $n \cdot \log^{O(1)} n$   
%\textcolour{red}{Why polynomial in \(n\)? Isn't the running time also
%polynomial in \(\log{n}\)?}  and space $O(\log n)$, and outputs a prime \(p > n\).
%We remark that it is alluded, but not explcitly stated that the
%AKS-primality test requires only ${\cal O}(\log q)$ space. 
%If this is
%true, 
Then we also have a deterministic algorithm to obtain a prime
number $p > n$ in time $n \cdot \log^{{\cal O}(1)} n$ and space ${\cal O}(\log n)$.
%This will de-randomise the process.

%We use the following lemma to obtain a lower-bound for the success
%probability of the algorithms that we design using colour coding.
%
%\begin{lemma}
%	\label{lm:successcount}
%	Fix a constant \(\epsilon \in (0, 1)\). Let \(G\) be the input graph and
%	\(S \subseteq V(G)\), such that \(| S | = k\).
%	Then the probability that \(| S |\) is properly coloured by
%	using randomly chosen functions from \(\mathcal{H}_u\), is at least \(1 - \epsilon\).
%%    \textcolour{red}{This statement cannot \textbf{possibly} be true! Please fix it.}
%\end{lemma}
%
%\begin{proof}
%	Let's pick \(\ceil{\text{log } \frac{1}{\epsilon}}\) functions
%	from \(\mathcal{H}_u\) at random.
%	As noted earlier, the
%	probablity that \(| S |\) is not colourful under a randomly
%	chosen colouring
%	function is at most \(1/2\). Thus, the probability that \(S\)
%	is not properly coloured under any of the randomly chosen
%	\(\ceil{\text{log } \frac{1}{\epsilon}}\) functions is at most \(\epsilon\). Hence, the above claim is
%	proved.
%\end{proof}

\textbf{Deleting vertices and edges:} Typically when deleting vertices
or edges from a graph, we make a copy of the given graph with those
vertices or edges deleted. However, in our setting we cannot simply
make a copy of a subgraph of the input graph, unless we can guarantee
its size to be bounded by some function of \(k\) only. If on the other
hand, as we shall encounter going forward, we can guarantee that the
size of the set of deleted vertices (or edges) can be bounded by a
function of \(k\) only, then we can \textit{simulate} the deletion,
i.e., we can make the algorithm act as if the set to be deleted has
actually been deleted. 

Let \(S \subseteq V\) be a set of vertices that
need to be deleted and \(| S | \leq \alpha(k)\), where \(\alpha\) is an
increasing function from \(\mathbb{N}\) to \(\mathbb{N}\). As its size
is bounded by a function of \(k\) only, so we can explicitly keep a
copy of \(S\) in our working space and mark it as deleted. Thus, if we
need to select a vertex from \(G - S\), we go through elements in \(V(G)\) and
check if they are from \(S\) or not.
%If any routine or a process or a loop 
%needs to consider \(S\) deleted, then all we do is that 
%%\textcolour{red}{This sentence is not clear. What are we trying to say here?} 
%process all vertex queries or operations on \(G\) except the ones in \(S\) and process all the
%edge queries or operations on \(G\) except the ones which have at least one endpoint
%in \(S\). 
Similarly, if \(F \subseteq E\) be a set of edges that needs to be
deleted, such that \(| F | \leq \beta(k)\), where \(\beta\) is an
increasing function from  \(\mathbb{N}\) to \(\mathbb{N}\), then we can simulate the 
deletion by storing a copy of \(F\) and marking it as deleted. If
\(G'\) is the subgraph of \(G\) which is obtained after deletion of
\(F\) from \(E\), then we can simply select edges of \(G'\) by
accessing the adjacency matrix of \(G\) via an oracle which checks if
and edge is in \(F\) or not. This can be clearly implemented in
\(\mathcal{O}(n^2)\) time and \(\mathcal{O}(\beta(k) \cdot \text{log } n)\)
space.
%not processing 
%%\textcolour{red}{What does it mean to ``not process'' a query? Could we spell this out? I think making this clearer will also address my previous comment.} 
%any edge queries or operations on edges from \(F\).

We will now proceed to describe the problems along with the
algorithms. 
%Details of the computations, including the pseudocodes and
%the proofs are provided in the respective sections of the appendix.

\section{\(k\)-PATH}

\noindent\fbox{\begin{minipage}{\textwidth}
Input: An undirected graph \(G = (V,E)\) and an
integer \(k\).	

Parameter: \(k\)

Question: Does \(G\) have a path on at least \(k\) vertices?
\end{minipage}}

%\begin{restatable}[]{theorem}{kpath}
%\label{th:kpath}
%	There is a constant probability, one-sided error, randomised algorithm that solves \(k\)-PATH
%	problem, which runs in time \(n^{\mathcal{O}(1)} \cdot 2^{k^2}
%	\cdot k!\) and uses 
%	\((\text{log } n)^{\mathcal{O}(1)} \cdot 2^{k^2} \cdot k!
%	\cdot k\) working space.
%\end{restatable}

\kpath*

%We will prove \cref{th:kpath} in detail in \cref{longpath}.
We will use the colour coding technique to design our algorithm. If
\(G\) is a {\sc Yes}-instance, then there exists some subset \(S\) of
\(V\) such that \(| S | = k\) and there is a path of length \(k - 1\)
on vertices of \(S\). As mentioned earlier, there exists some
colouring function \(h_{a,b}\) from the family \(\mathcal{H}_u\), such
that \(h_{a,b}\) will be injective when restricted to \(S\). Thus, we
can use the functions from \(\mathcal{H}_u\) one by one to identify
\(S\). And we have already seen that it can be done in
\(\mathcal{O}(n^3)\) time and \(\mathcal{O}(\log n)\) space.
We fix one such colouring function, and describe the algorithm with respect to it.
%By Lemma~\ref{lm:successcount}, we randomly pick a set of $\log \frac{1}{\epsilon}$ colouring functions from ${\cal H}_u$ and with probability at least $1-\epsilon$, one of them properly colours some $k$-path of the graph. 

\textbf{Description of the main algorithm:} \label{alg:mainkpath} We enumerate all
possible permutations of all possible subsets of \(k\) elements chosen from the set \(\{1, \ldots , k^2\}\). Let us consider one such permutation \(\{c_1, \ldots , c_k\}\).  
For each \(i \in [k]\), we define \(n_i =
|\{\text{Vertices of } G \text{ which have the colour } c_i\}|\). If for some \(i \in [k]\), \(n_i = 0\), then that is not a valid permutation and we move on to the next permutation of \(k\) colours.

Construct an auxiliary path \(\mathcal{P}\), such that
\(V(\mathcal{P}) = \{c_i \mid 1 \leq i \leq k\}\) and \(E(\mathcal{P}) = \{c_i c_{i + 1} \mid 1 \leq i \leq k - 1\}\) (a path on the \(k\) colours). We will use \cref{alg:findpath} to find a colourful path in \(G\) of size at least \(k\), if one such path exists in $G$.
%such that \(V(P) = \{w_i | 1 \leq i \leq k\}\), \(E(P) = \{w_i w_{i + 1} | 1 \leq i \leq k - 1\}\) and for each \(i \in [k]\), the colour of \(w_i\) is \(c_i\). 
If \cref{alg:findpath} returns {\sc Yes} for any input, then we return {\sc Yes}. Otherwise, if \cref{alg:findpath} return {\sc No} for all permutations of \(k\)-subsets of \([k^2]\), then we return {\sc No}.

In \cref{alg:findpath}, we (implicitly) construct an auxiliary graph $G^\star$ using the path $\cal P$ as follows: \(V(G^{\star}) \leftarrow \{s,t\} \cup V(G)\) where $s,t$ are two new vertices,  \(N(s) \leftarrow \{v \in V(G) \mid \text{colour of } v \text{ is } c_1\}\), \(N(t) \leftarrow \{v \in V(G) \mid \text{colour of } v \text{ is } c_k\}\),
and \(E(G^{\star}) \leftarrow \{sv \mid v \in N(s)\} \cup \{tv \mid v
\in N(t)\} \cup_{i = 1} ^{k - 1} \{vw \mid vw \in E(G) \land
\text{colour of \(v\) is } c_i \land \text{ colour of \(w\) is } c_{i + 1}\}\). 

Observe that, any path between $s$ and $t$ in $G^\star$, if one exists,  has at least $k$ internal vertices, by construction. Otherwise, if there is a path with fewer vertices, then as $s$ is adjacent to vertices of colour $c_1$ only and $t$ is adjacent to vertices of colour $c_k$ only, hence there must be an edge in this path between two vertices of colours $c_i$ and $c_j$ where $|i-j|\geq 2$, which contradicts the construction.

\begin{lemma}
\label{lm:nokpath1}
If \cref{alg:findpath} is correct and \(G\) is a {\sc No}-instance, then
the above algorithm will return {\sc No}.
\end{lemma}

\begin{proof}
	Suppose, the graph \(G\) is a {\sc No}-instance, i.e.
	we don't have a path of length at least \(k - 1\) in \(G\);
	which implies that we cannot have a colourful path of length
	at least \(k - 1\). Thus, for any permutation of any
	\(k\)-subset of the \(k^2\) colours, one cannot get any
	properly coloured path of length at least $k - 1$, starting
	from a vertex of the the first colour and ending at a vertex of
	the last colour. Thus, there is no path from $s$ to $t$ in
	$G^\star$. Hence, the algorithm \(\mathcal{A}_{con}\) will return {\sc No}. So our algorithm will correctly return {\sc No} as the answer.
	
\end{proof}

\begin{algorithm}
\caption{Finding a colourful path}\label{alg:findpath}
\begin{algorithmic}[1]
	\STATE \underline{\textbf{FindAPath(\(\mathcal{P}\))}}
	\STATE Add two vertices \(s\) and \(t\) which are not in the
	input graph \(G\).
	\STATE \(N(s) \leftarrow \{v \in V(G) \mid \text{colour of } v \text{ is } c_1\}\)
	\STATE \(N(t) \leftarrow \{v \in V(G) \mid \text{colour of } v \text{ is } c_k\}\)
	\STATE Construct an undirected graph \(G^{\star}\) as follows.
	\STATE \(V(G^{\star}) \leftarrow \{s,t\} \cup V(G)\)
	\STATE \(E(G^{\star}) \leftarrow \{sv \mid v \in N(s)\} \cup \{tv \mid v \in N(t)\} \cup_{i = 1} ^{k - 1} \{vw \mid vw \in E(G) \land \text{colour of v is } c_i \land \text{ colour of w is } c_{i + 1}\}\)
	\STATE Pass the information of \(G^{\star}\) to
	\(\mathcal{A}_{con}\) to check for the connectivity of \(s\)
	and \(t\)
	\STATE If \(s\) and \(t\) are connected in \(G^{\star}\), then
	return {\sc Yes}
	\STATE If \(s\) and \(t\) are not connected in \(G^{\star}\), then
	return {\sc No}
\end{algorithmic}
\end{algorithm}

Note that in \cref{alg:findpath}, the contruction of \(G^{\star}\) is not done
explicitly, as we do not have enough space for it. 
Instead, we provide access to their adjacency matrix via an oracle that can be implemented in $O(\log n)$-space and polynomial time. We first
introduce vertices \(s\) and \(t\) which are not already in \(V(G)\). Their neighbourhoods can be determined in \(\mathcal{O}(n)\) time and \(\mathcal{O}(\log n)\) space by scanning through \(V(G)\) and determining the colour of each vertex. As for determining the other edges, then note that
those vertices which do not have the colours in the input auxiliary path \(\mathcal{P}\) are considered isolated. And the vertices which have a colour, say \(c^{\star}\), from the set \(\{c_1, c_2, \ldots, c_{k - 1}\}\) then we only consider their neighbours which have the colour \(c^{\star} + 1\). 
%Vertices of colour \(c_k\) are reported to only have neighbour \(t\). 
It is with this oracle access to $G^\star$ that we call \(\mathcal{A}_{con}\) on \(G^{\star}, s\) and \(t\).

We prove the correctness of \cref{alg:findpath} in the following
lemma.

\begin{lemma}
	\label{lm:colourpath}
	There exists a path on \(k\)-vertices in \(G\), with the same
	colour configuration as \(\mathcal{P}\) if and only if, \(s\)
	and \(t\) are connected in \(G^{\star}\).
\end{lemma}

\begin{proof}
	Suppose, \(G\) has path \(P\) with the same colour
	configuration as \(\mathcal{P}\), i.e., \(V(P) = \{w_i \mid 1 \leq i \leq k\}\), \(E(P) =
\{w_i w_{i + 1} | 1 \leq i \leq k - 1\}\) and for each \(i \in [k]\),
the colour of \(w_i\) is \(c_i\). Then, by construction there exists
a path from \(s\) to \(t\) in \(G^{\star}\), as \(w_1 \in N(s)\) and
\(w_k \in N(t)\). Therefore, the algorithm \(\mathcal{A}_{con}\) on \(G^{\star}, s\) and \(t\) will return {\sc Yes}.

Conversely, suppose that \(s\) and \(t\) are connected in
\(G^{\star}\). Then, there exists a path between \(s\) and \(t\) in \(G^{\star}\), say \(P'\). Recall that, by the construction of $G^\star$, any path from $s$ to $t$ has at least $k$ internal vertices. Observe that, the internal vertices and edges of $P'$ are also present in $G$, and thus we obtain a path of length at least $k$ in $G$. 

%As the only neighbours of \(s\) are vertices of colour \(c_1\), then there exists a vertex, say \(w\), of colour \(c_1\) in \(P'\). Similarly, there exists a vertex of colour \(c_k\) in \(P'\). By construction, for a vertex of colour from the list \(\langle c_1, \ldots, c_{k - 1} \rangle\) the only neighbours are those with the next colour in the list. Thus, if there exists a vertex of colour, say \(c^{\star}\) from the set \(\{c_1, \ldots, c_{k - 1}\}\), in the path \(P'\), then there exists a vertex of \(c^{\star} + 1\) in \(P'\).Thus, \(P' - \{s, t\}\) is a path in \(G\) with the same colour configuration as \(\mathcal{P}\).

\end{proof}

\begin{lemma}
	\label{lm:tsp11}
	A single run of \cref{alg:findpath} takes
	\(n^{\mathcal{O}(1)}\) time and \(\mathcal{O}(k \cdot
	\log n)\) space.
\end{lemma}

\begin{proof}
	We know from \cite{10.1145/1060590.1060647}, that a call to
	\(\mathcal{A}_{con}\) will take up polynomial time and
	\(\mathcal{O}(\log n)\) space. Apart from that, the rest
	of \cref{alg:findpath} clearly takes up polynomial time. As
	for space then we need \(\mathcal{O}(k \cdot \log n)\)
	for the description of \(\mathcal{P}\), the vertices \(s\) and
	\(t\) and their neighbours.
\end{proof}

\begin{lemma}
	\label{lm:tsp12}
	The above main algorithm takes \(n^{\mathcal{O}(1)} \cdot
	2^{k^2} \cdot k!\) time and \(\mathcal{O}(2^{k^2} \cdot k!
	\cdot k \cdot \log n)\) space.
\end{lemma}

\begin{proof}
	There are most \(2^{k^2} \cdot k!\) possible permutations of
	\(k\)-elements chosen from the set \([k^2]\) and hence that
	many choices for the auxiliary path \(\mathcal{P}\). We list all of
	them. The space needed for that is \(2^{k^2} \cdot k! \cdot
	\log n\) and time is \(2^{k^2} \cdot k!\).

	For every choice of \(\mathcal{P}\), we need to call
	\cref{alg:findpath} at most once. By \cref{lm:tsp11}, the
	total time taken is \(n^{\mathcal{O}(1)} \cdot
	2^{k^2} \cdot k!\) and the space used it \(\mathcal{O}(2^{k^2} \cdot k!
	\cdot k \cdot \log n)\).
\end{proof}

%By \cref{lm:successcount}, for a constant \(\epsilon \in (0,1)\), one
%can make a \(k\)-path in \(G\) to be colourful using the family
%\(\mathcal{H}_u\), with probability at least \(1 - \epsilon\). 
We thus have a proof of \cref{th:kpath}.

\kpath*

\section{MAXLEAF SUBTREE}\label{maxleaf-l}

\noindent\fbox{\begin{minipage}{\textwidth}
Input: An undirected graph \(G = (V,E)\) and an integer \(k\).	

Parameter: \(k\)

Question: Does \(G\) have a subtree with at least \(k\) leaves?
\end{minipage}}

\maxleafsubtree*

We will adapt the algorithm used in \cite{kneis2011new} to our bounded
space setting. 
%The pesudcode along with other relevant details will be provided in \cref{maxleaf}.

Before we describe the algorithm, we will introduce certain terms from
\cite{kneis2011new}, which will be needed.

Given a rooted tree \(T\), we denote its root by \(root(T)\). The set
of leaves of \(T\) will be denoted by \(leaves(T)\). A tree
with \(k\) leaves will be called \(k\)-leaf tree. A non-leaf vertex of
a tree is called an \(inner ~ vertex\). For a graph \(G\) and a rooted
subtree \(T\), we call \(T\) to be \textit{inner-maximal} rooted tree
if for every inner vertex \(v\) of \(T\), \(N_{G}(v) \subseteq V(T)\).

\begin{itemize}
	\item For rooted trees \(T\) and \(T'\), we say that \(T'\)
		extends \(T\), denoted by \(T' \succeq T\), iff
		\(root(T') = root(T)\) and \(T\) is an induced
		subgraph of \(T'\). We write \(T' \succ T\), when \(T'
		\succeq T\) and \(T' \neq T\).

	\item Given a rooted tree \(T\), the algorithm distinguishes
		its leaves into two kinds. The \(red\) leaves \(R\) of
		\(T\), are those which will remain as leaves in any
		other tree which extends \(T\). The \(blue\) leaves
		\(B\) of \(T\) are those which may be inner vertices
		for some other tree which extends \(T\).
	\item A \textit{leaf-labelled tree} is \(3\)-tuple \((T, R, B)\),
		such that \(T\) is a rooted tree, \(R \cup B =
		leaves(T)\) and \(R \cap B = \emptyset\). 
		A leaf-labelled rooted tree \((T, R, B)\) is called
		\textit{inner-maximal leaf-labelled rooted tree}, if \(T\) is
		inner-maximal. If \((T, R,
		B)\) is a leaf-labelled tree and \(T'\) is a rooted
		tree such that \(T' \succeq T\) and \(R \subseteq
		leaves(T')\), we say that \(T'\) is a
		\textit{(leaf-preserving) extension} of \((T, R, B)\)
		denoted by \(T' \succeq (T, R, B)\). A leaf-labelled
		rooted tree \((T', R', B')\) \(extends\) a
		leaf-labelled rooted tree \((T, R, B)\) , denoted by
		\((T', R', B') \succeq (T, R, B)\), iff \(T' \succeq
		(T, R, B)\) and \(R \subseteq R'\).
\end{itemize}

We have following observations that will be needed to describe our
algorithm.

\begin{lemma}
	\label{lm:obs31}
	Suppose two inner-maximal trees, \(T\) and \(T'\) are such
	that \(root(T) = root(T')\) and \(leaves(T) = leaves(T')\).
	Then \(V(T) = V(T')\).
\end{lemma}

\begin{proof}
	Let \(u \in V(T)\). If \(u = root(T)\), then \(u \in V(T')\).
	If \(u \in leaves(T)\), then \(u \in leaves(T')\). Suppose,
	\(u\) is neither a root nor a leaf. Then there exists a path,
	\(v_1, \ldots , v_s\), in \(T\), such that \(v_1 = root(T)\)
	and \(v_s = u\). As \(v_1 = root(T')\), then there exists an
	\(i\) such that \(v_i \in V(T')\). As \(T'\) is inner maximal,
	so \(v_{i + 1} \in V(T')\). Thus, \(u \in V(T')\). So \(V(T)
	\subseteq V(T')\). Similarly, \(V(T') \subseteq V(T)\) and
	hence \(V(T) = V(T')\).
\end{proof}

\begin{remark}
	Suppose \((G, k)\) is a {\sc Yes}-instance. Then, \(G\) has an
	inner maximal subtree with at least \(k\) leaves.
\end{remark}

For each \(v \in V(G)\), we define \(T_v\) be the tree rooted at \(v\),
such that \(V(T_v) = N[v]\) and \(E(T_v) = \{vw | w \in N(w)\}\).
Notice that in order to describe \(T_v\), we need not store a copy of
the entire tree but suffice with just knowing the root.

\textbf{Description of the main algorithm:} We first check if there
exists a vertex of degree at least \(k\). If there is one such vertex,
we return {\sc Yes}. Otherwise, we are in the case where each vertex in
\(G\) has degree at most \(k\) and for every \(v \in V(G)\),
we call \cref{alg:maxleaf}, with the input \((v, \emptyset, N(v))\).
This will check if there exists a \(k\)-leaf rooted subtree in \(G\) with
\(v\) as its root.

\begin{lemma}
	\label{lm:main31}
	If \cref{alg:maxleaf} is correct, then the above main
	algorithm is correct as well.
\end{lemma}

\begin{proof}
	If there exists a vertex of degree with at least \(k\), say
	\(v\), then \(T_v\) is the required tree and we may return
	{\sc Yes}.
	
	Let's consider the case where the maximum degree in \(G\) is
	at most \(k - 1\). Suppose that there exists a rooted subtree of \(G\) with
	at least \(k\) leaves, say \(T\). 
	Let the root of \(T\) to be
	\(w\). Then \cref{alg:maxleaf} with return {\sc Yes} with the
	input \((w, \emptyset, N(w))\).
	In case, \(G\) is a {\sc No}-instance, then \cref{alg:maxleaf}
	will return {\sc No} for any root. Thus, the main algorithm is
	correct.
\end{proof}

In \cref{alg:extendtree}, given an inner-maximal tree and and one of its leaf as
input, we use \(\mathcal{A}_{con}\) to detect the neighbours of the
leaf not already in the tree.

\begin{algorithm}
\caption{Finding a rooted tree with many leaves}\label{alg:maxleaf}
\begin{algorithmic}[1]
	\STATE \underline{{\sc MaxLeaf(\(root(T), R, B\))}}

	\IF { \(| R | + | B | \geq k\)}
	\STATE return {\sc Yes}
	\ENDIF
	\IF {\(B = \emptyset\)}
	\STATE return {\sc No}
	\ENDIF
	\STATE Choose \(u \in B\) 
	\IF {{\sc MaxLeaf}\((root(T), R \cup \{u\}, B \setminus
	\{u\})\)}\label{alg:maxleaf:branch1}
	\STATE return {\sc Yes}\textit{\slash \slash The branch where
	\(u\) remains a leaf}
	\ENDIF
	\STATE \(B \leftarrow B \setminus \{u\}\)
	\STATE \(N \leftarrow\) {\sc ExtendTree}\((u, root(T), R \cup B \cup
	\{u\})\) \textit{\slash \slash Let \(N\) be set of neighbours of \(u\) outside of
	\(T\)}
	\STATE \(i \leftarrow 0\)
	\STATE \(w \leftarrow \text{The first member of } N\)
	\WHILE {\((| leaves(T) | \neq k)  \lor (i \neq | N |)\)}
	\label{alg:maxleaf:while1}
	\STATE Add \(w \in T\) as a neighbour of \(u\)
	\STATE \(w \leftarrow \text{Next member of } N \text{ in their order}\)
	\STATE \(i \leftarrow i + 1\)
	\ENDWHILE
	\IF {\(| leaves(T) | = k\)}
	\STATE return {\sc Yes}
	\ENDIF
	\WHILE {\(| N | = 1\)} \label{alg:maxleaf:while2}
	\STATE \textit{\slash \slash Follow paths}
	\STATE Let \(u\) be the unique element of \(N\)
	\STATE \(N' \leftarrow\) {\sc ExtendTree}\((u, root(T), R \cup B \cup
	\{u\})\) 
	\STATE \(i \leftarrow 0\)
	\STATE \(w \leftarrow \text{The first member of } N'\)
	\WHILE {\((| leaves(T) | \neq k)  \land (i \neq | N' |)\)}
	\STATE Add \(w \in T\) as a neighbour of \(u\)
	\STATE \(w \leftarrow \text{Next member of } N' \text{ in their order}\)
	\STATE \(i \leftarrow i + 1\)
	\ENDWHILE
	\IF {\(| leaves(T) | = k\)}
	\STATE return {\sc Yes}
	\ENDIF
	\STATE \(N \leftarrow N'\)
	\ENDWHILE
	\IF {\(N = \emptyset\)}
	\STATE return {\sc No}
	\ENDIF
	\STATE return {\sc MaxLeaf}\((root(T), R, B \cup N)\)
\end{algorithmic}
\end{algorithm}

\begin{algorithm}
\caption{Finding neighbours of a leaf to extend}\label{alg:extendtree}
\begin{algorithmic}[1]
	\STATE \underline{{\sc ExtendTree(\(u, root(T), leaves(T)\))}}
	\IF {\(u \notin leaves(T)\)}
	\STATE return ``Wrong Input"
	\ENDIF
	\STATE \(r \leftarrow root(T)\)
	\STATE \(N \leftarrow \emptyset\)
	\STATE Let \(G^*\) be the induced subgraph of \(G\) on \(V(G)
	\setminus leaves(T)\) 
	\STATE (We construct \(G^*\) implicitly)
	\FOR {\(w \in N(u) \setminus leaves(T)\)}
	\STATE Use \(\mathcal{A}_{con}\) to check for connectivity of
	\(w\) and \(r\) in \(G^*\)
	\IF {\(w\) and \(r\) are not connected in \(G^*\)}
	\STATE \textit{\slash \slash By inner-maximality of \(T\)}
	\STATE \(N \leftarrow N \cup \{w\}\)
	\ENDIF
	\ENDFOR
	\STATE return \(N\)
\end{algorithmic}
\end{algorithm}

\begin{lemma}
	\label{lm:obs32}
	\cref{alg:extendtree} is correct, if given an inner-maximal
	tree \(T\) as input.
\end{lemma}

\begin{proof}
	By \cref{lm:obs31}, by fixing a root and set of leaves, we
	can only have a unique set of vertices of the tree.
	Thus, we can detect the neighbours of the input leaf \(u\)
	which are not already in
	the inner-maximal tree just by using connectivity; as any
	vertex \(w \in N(u) \setminus leaves(T)\) is connected to
	\(u\) if and only if it is in \(V(T)\). Also note that there
	can be at most \(k - 1\) many choices for \(w\), thus allowing
	to detect the required neighbours within the restricted space.
\end{proof}

\begin{lemma}
	\label{lm:obs33}
	In \cref{alg:maxleaf}, if given an inner-maximal tree \(T\), as
	input, it will pass inner-maximal trees as input to recursive
	calls.
\end{lemma}

\begin{proof}
	Let the chosen blue vertex \(u\) have some neighbours outside
	of the input tree \(T\). There are two branches. In the
	branch, where \(u\) remains a leaf (and hence a red leaf), no
	change is done to structure of \(T\). Thus, the input to this
	branch is an inner-maximal tree. In the branch where \(u\)
	becomes an inner vertex, by \cref{lm:obs32} we can get all
	its neighbours outside the input tree \(T\). They are then
	added to \(T\) as neighbours of \(u\) and as blue leaves. This
	tree, which is an inner-maximal tree, is then passed as input to the recursive call.
\end{proof}

The while-loop at line \ref{alg:maxleaf:while1} of \cref{alg:maxleaf}
adds the neighbours of the vertex \(u\) not already in \(T\). If the
number of leaves of \(T\) becomes \(k\), then it stops and by using
the if-condition following it, the algorithm return {\sc Yes}. The
while-loop at line \ref{alg:maxleaf:while2} follows a path which
emanates from \(u\) and adds it to \(T\), until it comes across a
vertex which can have at least two children in \(T\).

Thus, \cref{alg:maxleaf} mimics the algorithm in
\cite{kneis2011new}. We refer to
\cite{kneis2011new}, for detailed proofs on the correctness
of the original algorithm.

\begin{lemma}
	\label{lm:obs341}
	The time taken for \cref{alg:maxleaf} is
	\(n^{\mathcal{O}(1)} \cdot 4^k\).
\end{lemma}

\begin{proof}
	In \cref{alg:extendtree}. polynomial time is used. Apart from
	that, the rest of \cref{alg:maxleaf} runs in polynomial time
	except for the time in the recursive calls. By Lemma \(7\) in
	\cite{kneis2011new}, the number of recursive calls is
	\(\mathcal{O}(4^k)\). Thus, the total time taken in a run of
	\cref{alg:maxleaf} is \(n^{\mathcal{O}(1)} \cdot 4^k\).
\end{proof}

As an immediate corollary to \cref{lm:obs341}, we have.
\begin{corollary}
	\label{lm:obs35}
	The main algorithm takes \(n^{\mathcal{O}(1)} \cdot 4^k\)
	time.
\end{corollary}

\begin{proof}
	As we call \cref{alg:maxleaf} with the input \((T_v,
	\emptyset, N(v))\), for each \(v \in V(G)\), so we need a
	total of \(n^{\mathcal{O}(1)} \cdot 4^k\) time.
\end{proof}

\begin{lemma}
	\label{lm:obs36}
	The working space used by \cref{alg:maxleaf} is
	\(\mathcal{O}(4^k \cdot k \cdot \log n)\).
\end{lemma}

\begin{proof}
	In order for us to track the recursive calls, we need to draw
	the entire computation tree and mark the current node. By
	Lemma 7 of \cite{kneis2011new},  there are
	\(\mathcal{O}(4^k)\) recursive calls. And for each node, we just need to need
	to keep in record the root and labeled leaves of the input
	tree. Also, we return {\sc Yes} as soon as the number of
	leaves become \(k\). Thus, we need to store \(\mathcal{O}(k \cdot \log
	n)\) amount of information per node of the computation tree of
	\cref{alg:maxleaf}. Thus, we have proved the claim of the
	lemma.
\end{proof}

Since we can re-use space for each call to \cref{alg:maxleaf}, the
main algorithm uses \(\mathcal{O}(4^k \cdot k \cdot \log n)\)
working space.

We thus have a proof of \cref{th:maxleafsubtree}.

\maxleafsubtree*

\section{MULTICUT IN TREES}\label{multicuts}

\noindent\fbox{\begin{minipage}{\textwidth}
Input: An undirected tree \(T = (V,E), n = | V |\), a collection \(H\)
of \(m\) pairs of nodes in \(T\) and an integer \(k\).	

Parameter: \(k\)

Question: Does \(T\) have an edge subset of size at most \(k\) whose
removal separates each pair of nodes in \(H\)?
\end{minipage}}

\multicutintree*

We can assume that the elements of \(H\) are in order. Let \(H =
\{(a_i, b_i) | 1 \leq i \leq m\}\).

We will adapt the algorithm used in \cite{guo2005fixed} to our bounded
space setting. 
%The pseudocode and other relevant details are
%\cref{multicut}.

Before running the algorithm, we root \(T\) at an arbitrary
vertex, say \(r\). For two vertices \(u\) and \(v\), their
\textit{least common ancestor} is a vertex \(w\), such that it is an
ancestor to both \(u\) and \(v\) and has the greatest depth in \(T\)
measured by the distance from the root. Let \(E^{\star}\) be the solution we will be constructing and
initialise it to \(\emptyset\). At first we check if deleting the
edges already in \(E^{\star}\), will separate all the pairs in \(H\).
If so, then we return {\sc Yes} or else we identify the pair with the
deepest least common ancestor.
Once we have identified the pair with the
deepest least common ancestor with respect to the root, say \(\{u,
v\}\), we
check if the lowest common
ancestor, say \(w\), is one of the pair \(\{u,v\}\). If it is so, then we
delete the edge incident on \(w\) and lying in the uniquely
determined path between \(u\) and \(v\) and recursively call
the function with the appropriate values. Otherwise, there are
two edges incident on \(w\) from the path between \(u\) and
\(v\). We identify them and branch off to two recursive calls,
with each call representing the deletion of one such edge.

\begin{algorithm}
\caption{Finding Least Common Ancestor}\label{alg:lca}
\begin{algorithmic}[1]
	\STATE \underline{\textbf{FindLCA(\((a, b)\))}}
	\STATE \(lca \leftarrow r\)
	\STATE \textit{\slash \slash In the following lines, use
	\(\mathcal{A}_{con}\) to check for connectivity}
	\WHILE {\(lca\) is connected to both \(a\) and
	\(b\)}\label{alg:lca:whileloop}
	\STATE \(T' \leftarrow T - lca \) \textit{\slash \slash
	implicit deletion}
	\STATE \(j \leftarrow 0 \)
	\FOR {\(x \in N(lca)\)}\label{alg:lca:forloop}
	\IF {\(x\) is connected to both \(a\) and \(b\) in \(T'\)}
	\STATE \textit{\slash \slash Such an \(x\) will be unique as
	\(T\) is a tree}
	\STATE \(lca \leftarrow x\)
	\STATE \(j \leftarrow 1 \)
	\STATE Break the for-loop
	\ENDIF
	\ENDFOR
	\IF {\(j = 0\)} 
	\STATE \textit{\slash \slash The current \(lca\) has no child
	which is conected to both \(a\) and \(b\)}
	\STATE Break the while-loop
	\ENDIF
	\ENDWHILE
	\STATE return \(lca\)
\end{algorithmic}
\end{algorithm}

\begin{lemma}
	\label{lm:lcap}
	\cref{alg:lca} finds the least common ancestor for the input
	pair \((a, b)\).
\end{lemma}

\begin{proof}
	By definition the least common ancestor, or \textit{lca} in
	short, must be connected to both \(a\) and \(b\) to begin
	with. Thus, it makes sense to initialse the variable \(lca\)
	to \(r\). If the current \(lca\) is deleted from \(T\),
	then there exists at
	most one vertex in the resultant graph which is
	connected to both \(a\) and \(b\), i.e., a neighbour of
	\(lca\) before deletion. If such a vertex does exist, then we
	assign this vertex to the variable \(lca\). If such a vertex
	doesn't exist, then the current \(lca\) is indeed the least
	common ancestor of \(a\) and \(b\). The for-loop at line
	\ref{alg:lca:forloop} tries to find such a vertex. If such a
	vertex is found, then we break the while-loop at \ref{alg:lca:whileloop}.
\end{proof}

\begin{algorithm}[H]
\caption{Finding the distance from the root \(r\) of tree \(T\)}\label{alg:dist}
\begin{algorithmic}[1]
	\STATE \underline{\textbf{dist(\(x\))}}
	\STATE \textit{\slash \slash In the following lines, use
	\(\mathcal{A}_{con}\) to check for connectivity}
	\IF {\( x = r\)}
	\STATE return \(0\)
	\ENDIF
	\IF {\( x \in N(r)\)}
	\STATE return \(1\)
	\ENDIF
	\STATE \(d \leftarrow 0\)
	\STATE \(y \leftarrow r\)
	\WHILE {\(y \neq x\)}
	\STATE \(T' \leftarrow T - y\) \textit{\slash \slash
	implicit deletion}
	\WHILE {\(z \in N(y)\) is not connected to \(x\) in \(T'\)}
	\STATE Move to next member in \(N(y)\)
	\ENDWHILE
	\STATE \textit{\slash \slash There is a unique \(z \in N(y)\), which is
	connected to \(x\) in \(T'\), as \(T\) is a tree}
	\STATE \(y \leftarrow z\)
	\STATE \(d \leftarrow d + 1\)
	\ENDWHILE
	\STATE return \(d\)
\end{algorithmic}
\end{algorithm}

\begin{lemma}
	\label{lm:distp}
	\cref{alg:dist} finds the distance of the vertex \(x\) from
	the root \(r\) of \(T\).
\end{lemma}

\begin{proof}
	If \(x = r\) or \(x \in N(r)\), then the algorithm correctly
	provides the answer. Otherwise, it uses the variable \(y\) to
	determine the current vertex from which the distance to \(x\)
	needs to evaluated. Initially, \(y\) is set to \(r\). After
	that we determine the unique vertex lying in the path from
	\(y\) to \(x\). This is exactly the unique neighbour of \(y\), which
	is connected to \(x\) in \(T - y\). We keep a count of such
	vertices until we reach \(x\) and increase the count by one
	for each such vertex. Then we return the value of \(d\) as
	distance of \(x\) from \(r\).
\end{proof}

\begin{algorithm}
\caption{Finding Multicut}\label{alg:multicut}
\begin{algorithmic}[1]
	\STATE \underline{\textbf{Multicut(\((E^{\star})\))}}
	\STATE Let \(T^{\star}\) be the subgraph of \(T\) obtained
	after deleting \(E^{\star}\) from the set of edges in \(T\)
	\STATE \(count \leftarrow 0\)
	\STATE \(d \leftarrow - 1\)
	\STATE \(i \leftarrow 0\)
	\FOR {\(i \in [m]\)} \label{alg:multicut:linefor}
	\STATE Use \(\mathcal{A}_{con}\) to check if \(a_i\) and
	\(b_i\) are connected in \(T^{\star}\)
	\IF {\(a_i\) and \(b_i\) are connected in \(T^{\star}\)}
	\IF {\( | E^{\star} | = k\)}
	\STATE return \(NO\)\label{alg:multicut:no}
	\ENDIF
	\STATE \(w \leftarrow FindLCA(a_i, b_i)\) \textit{\slash
		\slash Find the least common ancestor of \(a_i\) and
	\(b_i\)}
	\IF{\(dist(w) > d\)}
	\STATE \(u \leftarrow a_i\)
	\STATE \(v \leftarrow b_i\)
	\STATE \(d \leftarrow dist(w)\)
	\ENDIF
	\ELSE
	\STATE \(count \leftarrow count + 1\)
	\ENDIF
	\STATE \(i \leftarrow i + 1\)
	\ENDFOR
	\IF {\(count = m\)}
	\STATE \textit{\slash \slash All pairs are separated}
	\STATE return \(YES\)\label{alg:multicut:yes} 
	\ENDIF
	\IF {\((w = u) \lor (w = v)\)}
	\STATE Let \(T'\) be the subgraph obtained after deleting
	\(w\) from \(T\) (implicitly)
	\STATE Let \( x \in \{u, v\} \setminus \{w\}\) \textit{\slash
	\slash The other member of the pair}
	\FOR {\(y \in N(w)\)}
	\STATE Use \(\mathcal{A}_{con}\) to check if \(x\) and \(y\)
	are connected in \(T'\)
	\IF {\(x\) and \(y\) are connected in \(T'\)}
	\STATE \(e \leftarrow yw\)
	\STATE Break the for-loop
	\ENDIF
	\ENDFOR
	\STATE return \(Multicut(E^{\star} \cup \{e\})\)
	\ELSE
	\FOR {\( x \in \{u, v\}\)}
	\FOR {\(y \in N(w)\)}
	\STATE Use \(\mathcal{A}_{con}\) to check if \(x\) and \(y\)
	are connected in \(T'\)
	\IF {\(x\) and \(y\) are connected in \(T'\)}
	\STATE \(e_x \leftarrow yw\)
	\STATE Break the inner for-loop
	\ENDIF
	\ENDFOR
	\ENDFOR
	\ENDIF
	\STATE return \(Multicut(E^{\star} \cup \{e_u\}) \lor Multicut(E^{\star} \cup \{e_v\})\)
\end{algorithmic}
\end{algorithm}

\begin{remark}
	In line \ref{alg:multicut:linefor} of \cref{alg:multicut}, the
	for-loop checks if by deleting \(E^{\star}\) from \(T\), one
	can still have unseparated pairs in \(H\). If not, then it
	returns {\sc Yes} (line \ref{alg:multicut:yes}) or else it moves on the rest of the
	algorithm having detected the pair with the deepest
	least common ancestor.
\end{remark}

\begin{lemma}
	\label{lm:obs43}
	The size of \(E^{\star}\) doesn't exceed \(k\) in a run of
	\cref{alg:multicut}.
\end{lemma}

\begin{proof}
	\(E^{\star}\) is initially \(\emptyset\). In every recursive call, exactly one edge
	is added and when \(|E^{\star} |\) reaches \(k\), either {\sc
	Yes}(line \ref{alg:multicut:yes}) or {\sc No}(line \ref{alg:multicut:no}) is
	returned.
\end{proof}

With the above remarks, we can conclude that
\cref{alg:multicut} mimics the algorithm in Theorem \(1\) in
\cite{guo2005fixed}. We refer to \cite{guo2005fixed}, for details on
the correctness proofs of the original algorithm.

\begin{lemma}
	\label{lm:tsp41}
	\cref{alg:multicut} takes \(n^{\mathcal{O}(1)} \cdot 2^k\)
	time.	
\end{lemma} 

\begin{proof}
	In a call to \cref{alg:multicut}, there are at most \(m\)
	calls to \cref{alg:lca} and \cref{alg:dist}, each of which
	take \(n^{\mathcal{O}(1)}\) time. Also, by Theorem 1 of
	\cite{guo2005fixed}, there are atmost \(2^k\) recursive calls.
	Thus, the total time taken is \(n^{\mathcal{O}(1)} \cdot
	2^k\).
\end{proof}

\begin{remark}
	\label{lm:tsp421}
	\cref{alg:lca} and \cref{alg:dist} take
	\(\mathcal{O}(\log n)\) working space.
\end{remark}
%\todo[inline]{Add a proof}

\begin{lemma}
	\label{lm:tsp43}
	\cref{alg:multicut} takes \(\mathcal{O}(2^k \cdot k \cdot
	\log n)\) working space.
\end{lemma}

\begin{proof}
	By \cref{lm:tsp421}, a single node of the computation tree of \cref{alg:multicut} will
	take \(\mathcal{O}(k \cdot \log n)\) space. And as there are at most
	\(2^k\) recursive calls, the total working space needed is \(\mathcal{O}(2^k \cdot k \cdot
	\log n)\).
\end{proof}

We have thus proved \cref{th:multicutintree}.

\multicutintree*

\section{Conclusion}

We discussed three graph theoretic problems in the space bounded
settings and provided {\sc Fpt}-time algorithms for them, namely
\(k\)-{\sc Path}, {\sc MaxLeaf Subtree} and {\sc Multicut} in Trees.
It would be interesting
to see whether other standard graph theoretic problems can also admit
{\sc Fpt} algorithms when the working space is bounded, like the {\sc
Steiner Tree} problem. Also, improving the running time of the
problems already solved here can be an interesting area of study.

\bibliographystyle{splncs04}
\bibliography{references}

@book{arora2009computational,
  title={Computational complexity: a modern approach},
  author={Arora, Sanjeev and Barak, Boaz},
  year={2009},
  publisher={Cambridge University Press}
}

@article{ElberfeldST15,
  author       = {Michael Elberfeld and
                  Christoph Stockhusen and
                  Till Tantau},
  title        = {On the Space and Circuit Complexity of Parameterized Problems: Classes
                  and Completeness},
  journal      = {Algorithmica},
  volume       = {71},
  number       = {3},
  pages        = {661--701},
  year         = {2015},
  url          = {https://doi.org/10.1007/s00453-014-9944-y},
  doi          = {10.1007/S00453-014-9944-Y},
  bibsource    = {dblp computer science bibliography, https://dblp.org}
}

@inproceedings{BodlaenderWG24,
  author       = {Hans L. Bodlaender and
                  Krisztina Szil{\'{a}}gyi},
  editor       = {Daniel Kr{\'{a}}l and
                  Martin Milanic},
  title        = {XNLP-Hardness of Parameterized Problems on Planar Graphs},
  booktitle    = {Graph-Theoretic Concepts in Computer Science - 50th International
                  Workshop, {WG} 2024, Gozd Martuljek, Slovenia, June 19-21, 2024, Revised
                  Selected Papers},
  series       = {Lecture Notes in Computer Science},
  volume       = {14760},
  pages        = {107--120},
  publisher    = {Springer},
  year         = {2024},
  url          = {https://doi.org/10.1007/978-3-031-75409-8\_8},
  doi          = {10.1007/978-3-031-75409-8\_8},
  bibsource    = {dblp computer science bibliography, https://dblp.org}
}

@inproceedings{BodlaenderIPEC22,
  author       = {Hans L. Bodlaender and
                  Carla Groenland and
                  Hugo Jacob and
                  Lars Jaffke and
                  Paloma T. Lima},
  editor       = {Holger Dell and
                  Jesper Nederlof},
  title        = {XNLP-Completeness for Parameterized Problems on Graphs with a Linear
                  Structure},
  booktitle    = {17th International Symposium on Parameterized and Exact Computation,
                  {IPEC} 2022, September 7-9, 2022, Potsdam, Germany},
  series       = {LIPIcs},
  volume       = {249},
  pages        = {8:1--8:18},
  publisher    = {Schloss Dagstuhl - Leibniz-Zentrum f{\"{u}}r Informatik},
  year         = {2022},
  url          = {https://doi.org/10.4230/LIPIcs.IPEC.2022.8},
  doi          = {10.4230/LIPICS.IPEC.2022.8},
  bibsource    = {dblp computer science bibliography, https://dblp.org}
}

@inproceedings{BodlaenderESA22,
  author       = {Hans L. Bodlaender and
                  Carla Groenland and
                  Hugo Jacob},
  editor       = {Shiri Chechik and
                  Gonzalo Navarro and
                  Eva Rotenberg and
                  Grzegorz Herman},
  title        = {List Colouring Trees in Logarithmic Space},
  booktitle    = {30th Annual European Symposium on Algorithms, {ESA} 2022, September
                  5-9, 2022, Berlin/Potsdam, Germany},
  series       = {LIPIcs},
  volume       = {244},
  pages        = {24:1--24:15},
  publisher    = {Schloss Dagstuhl - Leibniz-Zentrum f{\"{u}}r Informatik},
  year         = {2022},
  url          = {https://doi.org/10.4230/LIPIcs.ESA.2022.24},
  doi          = {10.4230/LIPICS.ESA.2022.24},
  bibsource    = {dblp computer science bibliography, https://dblp.org}
}

@inproceedings{BodlaenderFOCS21,
  author       = {Hans L. Bodlaender and
                  Carla Groenland and
                  Jesper Nederlof and
                  C{\'{e}}line M. F. Swennenhuis},
  title        = {Parameterized Problems Complete for Nondeterministic {FPT} time and
                  Logarithmic Space},
  booktitle    = {62nd {IEEE} Annual Symposium on Foundations of Computer Science, {FOCS}
                  2021, Denver, CO, USA, February 7-10, 2022},
  pages        = {193--204},
  publisher    = {{IEEE}},
  year         = {2021},
  url          = {https://doi.org/10.1109/FOCS52979.2021.00027},
  doi          = {10.1109/FOCS52979.2021.00027},
  bibsource    = {dblp computer science bibliography, https://dblp.org}
}

@inproceedings{10.1007/978-3-031-22105-7_23,
	author = {Chen, Jianer and Chu, Zirui and Guo, Ying and Yang, Wei},
	title = {Space Limited Graph Algorithms on Big Data},
	year = {2023},
	isbn = {978-3-031-22104-0},
	publisher = {Springer-Verlag},
	address = {Berlin, Heidelberg},
	url = {https://doi.org/10.1007/978-3-031-22105-7_23},
	doi = {10.1007/978-3-031-22105-7_23},
	booktitle = {Computing and Combinatorics: 28th International Conference, COCOON 2022, Shenzhen, China, October 22–24, 2022, Proceedings},
	pages = {255–267},
	numpages = {13},
	location = {Shenzhen, China}
}

@InProceedings{10.1007/978-981-97-2340-9_22,
	author="Kammer, Frank
	and Sajenko, Andrej",
	editor="Chen, Xujin
	and Li, Bo",
	title="Space-Efficient Graph Kernelizations",
	booktitle="Theory and Applications of Models of Computation",
	year="2024",
	publisher="Springer Nature Singapore",
	address="Singapore",
	pages="260--271",
	isbn="978-981-97-2340-9"
}

@article{kneis2011new,
	title={A new algorithm for finding trees with many leaves},
	author={Kneis, Joachim and Langer, Alexander and Rossmanith, Peter},
	journal={Algorithmica},
	volume={61},
	pages={882--897},
	year={2011},
	publisher={Springer}
}

@article{guo2005fixed,
	title={Fixed-parameter tractability and data reduction for multicut in trees},
	author={Guo, Jiong and Niedermeier, Rolf},
	journal={Networks: An International Journal},
	volume={46},
	number={3},
	pages={124--135},
	year={2005},
	publisher={Wiley Online Library}
}

@InProceedings{10.1007/978-3-642-11269-0_1,
author="Alon, Noga
and Gutner, Shai",
editor="Chen, Jianer
and Fomin, Fedor V.",
title="Balanced Hashing, Color Coding and Approximate Counting",
booktitle="Parameterized and Exact Computation",
year="2009",
publisher="Springer Berlin Heidelberg",
address="Berlin, Heidelberg",
pages="1--16",
isbn="978-3-642-11269-0"
}

@article{b68c33ca-3366-3b13-8901-69e76cc88da6,
 ISSN = {0003486X},
 URL = {http://www.jstor.org/stable/3597229},
 author = {Manindra Agrawal and Neeraj Kayal and Nitin Saxena},
 journal = {Annals of Mathematics},
 number = {2},
 pages = {781--793},
 publisher = {Annals of Mathematics},
 title = {PRIMES Is in P},
 urldate = {2024-09-24},
 volume = {160},
 year = {2004}
}

@article{CARTER1979143,
title = {Universal classes of hash functions},
journal = {Journal of Computer and System Sciences},
volume = {18},
number = {2},
pages = {143-154},
year = {1979},
issn = {0022-0000},
doi = {https://doi.org/10.1016/0022-0000(79)90044-8},
url = {https://www.sciencedirect.com/science/article/pii/0022000079900448},
author = {J.Lawrence Carter and Mark N. Wegman}
}

@book{10.5555/1614191,
author = {Cormen, Thomas H. and Leiserson, Charles E. and Rivest, Ronald L. and Stein, Clifford},
title = {Introduction to Algorithms, Third Edition},
year = {2009},
isbn = {0262033844},
publisher = {The MIT Press},
edition = {3rd}
}

@inproceedings{10.1145/1060590.1060647,
author = {Reingold, Omer},
title = {Undirected ST-connectivity in log-space},
year = {2005},
isbn = {1581139608},
publisher = {Association for Computing Machinery},
address = {New York, NY, USA},
url = {https://doi.org/10.1145/1060590.1060647},
doi = {10.1145/1060590.1060647},
booktitle = {Proceedings of the Thirty-Seventh Annual ACM Symposium on Theory of Computing},
pages = {376–385},
numpages = {10},
location = {Baltimore, MD, USA},
series = {STOC '05}
}

@inproceedings{10.1145/237814.237823,
author = {Alon, Noga and Matias, Yossi and Szegedy, Mario},
title = {The space complexity of approximating the frequency moments},
year = {1996},
isbn = {0897917855},
publisher = {Association for Computing Machinery},
address = {New York, NY, USA},
url = {https://doi.org/10.1145/237814.237823},
doi = {10.1145/237814.237823},
booktitle = {Proceedings of the Twenty-Eighth Annual ACM Symposium on Theory of Computing},
pages = {20–29},
numpages = {10},
location = {Philadelphia, Pennsylvania, USA},
series = {STOC '96}
}

@InProceedings{chitnis_et_al:LIPIcs.IPEC.2019.7,
  author =	{Chitnis, Rajesh and Cormode, Graham},
  title =	{{Towards a Theory of Parameterized Streaming Algorithms}},
  booktitle =	{14th International Symposium on Parameterized and Exact Computation (IPEC 2019)},
  pages =	{7:1--7:15},
  series =	{Leibniz International Proceedings in Informatics (LIPIcs)},
  ISBN =	{978-3-95977-129-0},
  ISSN =	{1868-8969},
  year =	{2019},
  volume =	{148},
  editor =	{Jansen, Bart M. P. and Telle, Jan Arne},
  publisher =	{Schloss Dagstuhl -- Leibniz-Zentrum f{\"u}r Informatik},
  address =	{Dagstuhl, Germany},
  URL =		{https://drops.dagstuhl.de/entities/document/10.4230/LIPIcs.IPEC.2019.7},
  URN =		{urn:nbn:de:0030-drops-114682},
  doi =		{10.4230/LIPIcs.IPEC.2019.7}
}

@InProceedings{ghosh_et_al:LIPIcs.ESA.2024.60,
  author =	{Ghosh, Prantar and Kuchlous, Sahil},
  title =	{{New Algorithms and Lower Bounds for Streaming Tournaments}},
  booktitle =	{32nd Annual European Symposium on Algorithms (ESA 2024)},
  pages =	{60:1--60:19},
  series =	{Leibniz International Proceedings in Informatics (LIPIcs)},
  ISBN =	{978-3-95977-338-6},
  ISSN =	{1868-8969},
  year =	{2024},
  volume =	{308},
  editor =	{Chan, Timothy and Fischer, Johannes and Iacono, John and Herman, Grzegorz},
  publisher =	{Schloss Dagstuhl -- Leibniz-Zentrum f{\"u}r Informatik},
  address =	{Dagstuhl, Germany},
  URL =		{https://drops.dagstuhl.de/entities/document/10.4230/LIPIcs.ESA.2024.60},
  URN =		{urn:nbn:de:0030-drops-211318},
  doi =		{10.4230/LIPIcs.ESA.2024.60}
}

@book{DBLP:books/sp/CyganFKLMPPS15,
  author       = {Marek Cygan and
                  Fedor V. Fomin and
                  Lukasz Kowalik and
                  Daniel Lokshtanov and
                  D{\'{a}}niel Marx and
                  Marcin Pilipczuk and
                  Michal Pilipczuk and
                  Saket Saurabh},
  title        = {Parameterized Algorithms},
  publisher    = {Springer},
  year         = {2015},
  url          = {https://doi.org/10.1007/978-3-319-21275-3},
  doi          = {10.1007/978-3-319-21275-3},
  isbn         = {978-3-319-21274-6},
  bibsource    = {dblp computer science bibliography, https://dblp.org}
}

@inproceedings{DBLP:conf/mfcs/FafianieK15,
  author       = {Stefan Fafianie and
                  Stefan Kratsch},
  editor       = {Giuseppe F. Italiano and
                  Giovanni Pighizzini and
                  Donald Sannella},
  title        = {A Shortcut to (Sun)Flowers: Kernels in Logarithmic Space or Linear
                  Time},
  booktitle    = {Mathematical Foundations of Computer Science 2015 - 40th International
                  Symposium, {MFCS} 2015, Milan, Italy, August 24-28, 2015, Proceedings,
                  Part {II}},
  series       = {Lecture Notes in Computer Science},
  volume       = {9235},
  pages        = {299--310},
  publisher    = {Springer},
  year         = {2015},
  url          = {https://doi.org/10.1007/978-3-662-48054-0\_25},
  doi          = {10.1007/978-3-662-48054-0\_25},
  bibsource    = {dblp computer science bibliography, https://dblp.org}
}

@article{CHEN2022104951,
title = {Linear-time parameterized algorithms with limited local resources},
journal = {Information and Computation},
volume = {289},
pages = {104951},
year = {2022},
issn = {0890-5401},
doi = {https://doi.org/10.1016/j.ic.2022.104951},
url = {https://www.sciencedirect.com/science/article/pii/S0890540122001067},
author = {Jianer Chen and Ying Guo and Qin Huang}
}

@article{FEIGENBAUM2005207,
title = {On graph problems in a semi-streaming model},
journal = {Theoretical Computer Science},
volume = {348},
number = {2},
pages = {207-216},
year = {2005},
note = {Automata, Languages and Programming: Algorithms and Complexity (ICALP-A 2004)},
issn = {0304-3975},
doi = {https://doi.org/10.1016/j.tcs.2005.09.013},
url = {https://www.sciencedirect.com/science/article/pii/S0304397505005323},
author = {Joan Feigenbaum and Sampath Kannan and Andrew McGregor and Siddharth Suri and Jian Zhang}
}

@inbook{doi:10.1137/1.9781611977912.28,
author = {Daniel Lokshtanov and Pranabendu Misra and Fahad Panolan and M. S. Ramanujan and Saket Saurabh and Meirav Zehavi},
title = {Meta-theorems for Parameterized Streaming Algorithms‡},
booktitle = {Proceedings of the 2024 Annual ACM-SIAM Symposium on Discrete Algorithms (SODA)},
chapter = {},
pages = {712-739},
doi = {10.1137/1.9781611977912.28},
URL = {https://epubs.siam.org/doi/abs/10.1137/1.9781611977912.28},
eprint = {https://epubs.siam.org/doi/pdf/10.1137/1.9781611977912.28}
}

@InProceedings{bergougnoux_et_al:LIPIcs.ESA.2023.18,
  author =	{Bergougnoux, Benjamin and Chekan, Vera and Ganian, Robert and Kant\'{e}, Mamadou Moustapha and Mnich, Matthias and Oum, Sang-il and Pilipczuk, Micha{\l} and van Leeuwen, Erik Jan},
  title =	{{Space-Efficient Parameterized Algorithms on Graphs of Low Shrubdepth}},
  booktitle =	{31st Annual European Symposium on Algorithms (ESA 2023)},
  pages =	{18:1--18:18},
  series =	{Leibniz International Proceedings in Informatics (LIPIcs)},
  ISBN =	{978-3-95977-295-2},
  ISSN =	{1868-8969},
  year =	{2023},
  volume =	{274},
  editor =	{G{\o}rtz, Inge Li and Farach-Colton, Martin and Puglisi, Simon J. and Herman, Grzegorz},
  publisher =	{Schloss Dagstuhl -- Leibniz-Zentrum f{\"u}r Informatik},
  address =	{Dagstuhl, Germany},
  URL =		{https://drops.dagstuhl.de/entities/document/10.4230/LIPIcs.ESA.2023.18},
  URN =		{urn:nbn:de:0030-drops-186710},
  doi =		{10.4230/LIPIcs.ESA.2023.18}
}

@article{10.1145/2627692.2627694,
author = {McGregor, Andrew},
title = {Graph stream algorithms: a survey},
year = {2014},
issue_date = {March 2014},
publisher = {Association for Computing Machinery},
address = {New York, NY, USA},
volume = {43},
number = {1},
issn = {0163-5808},
url = {https://doi.org/10.1145/2627692.2627694},
doi = {10.1145/2627692.2627694},
journal = {SIGMOD Rec.},
month = may,
pages = {9–20},
numpages = {12}
}

@inproceedings{10.1145/195058.195179,
author = {Alon, Noga and Yuster, Raphy and Zwick, Uri},
title = {Color-coding: a new method for finding simple paths, cycles and other small subgraphs within large graphs},
year = {1994},
isbn = {0897916638},
publisher = {Association for Computing Machinery},
address = {New York, NY, USA},
url = {https://doi.org/10.1145/195058.195179},
doi = {10.1145/195058.195179},
booktitle = {Proceedings of the Twenty-Sixth Annual ACM Symposium on Theory of Computing},
pages = {326–335},
numpages = {10},
location = {Montreal, Quebec, Canada},
series = {STOC '94}
}

@article{10.1093/bioinformatics/btn163,
author = {Alon, Noga and Dao, Phuong and Hajirasouliha, Iman and Hormozdiari, Fereydoun and Sahinalp, S. Cenk},
title = {Biomolecular network motif counting and discovery by color coding},
year = {2008},
issue_date = {July 2008},
publisher = {Oxford University Press, Inc.},
address = {USA},
volume = {24},
number = {13},
issn = {1367-4803},
url = {https://doi.org/10.1093/bioinformatics/btn163},
doi = {10.1093/bioinformatics/btn163},
journal = {Bioinformatics},
month = jul,
pages = {i241–i249}
}

@article{10.1109/TCBB.2020.3040910,
author = {Davidov, Nathan and Hernandez, Amanda and Jian, Justin and McKenna, Patrick and Medlin, K.A. and Mojumder, Roadra and Owen, Megan and Quijano, Andrew and Rodriguez, Amanda and St. John, Katherine and Thai, Katherine and Uraga, Meliza},
title = {Maximum Covering Subtrees for Phylogenetic Networks},
year = {2020},
issue_date = {Nov.-Dec. 2021},
publisher = {IEEE Computer Society Press},
address = {Washington, DC, USA},
volume = {18},
number = {6},
issn = {1545-5963},
url = {https://doi.org/10.1109/TCBB.2020.3040910},
doi = {10.1109/TCBB.2020.3040910},
journal = {IEEE/ACM Trans. Comput. Biol. Bioinformatics},
month = nov,
pages = {2823–2827},
numpages = {5}
}

@inproceedings{10.5555/1873601.1873635,
author = {Barman, Siddharth and Chawla, Shuchi},
title = {Region growing for multi-route cuts},
year = {2010},
isbn = {9780898716986},
publisher = {Society for Industrial and Applied Mathematics},
address = {USA},
booktitle = {Proceedings of the Twenty-First Annual ACM-SIAM Symposium on Discrete Algorithms},
pages = {404–418},
numpages = {15},
location = {Austin, Texas},
series = {SODA '10}
}

%\newpage
%\appendix
%\section*{APPENDIX}
%
%\input{appendix}
%
%
%\input{maxleaf-l}
%\input{multicut-l}

%\input{long-path}
%\input{steiner-tree}
%\input{maxleaf}
%\input{multicut}

\end{document}